\documentclass[11pt]{article}

\usepackage{array}
\usepackage{amsfonts,amsmath,amssymb,amsthm,mathrsfs}
\usepackage{adjustbox}
\usepackage{tikz}
\usepackage{enumitem}
\usepackage{ifthen}
\usepackage{fullpage}

\usepackage[ruled]{algorithm2e} 

\usepackage{algorithmic}
\SetAlFnt{\small}
\SetAlCapFnt{\small}
\SetAlCapNameFnt{\small}
\SetAlCapHSkip{0pt}
\IncMargin{-\parindent}

\usepackage{booktabs} 
\usepackage{xspace}
\usepackage[pagebackref=true,hidelinks]{hyperref}
\usepackage{bbm}
\usepackage{enumitem}
\usepackage[numbers]{natbib}
\usepackage{multirow}
\usepackage{graphicx}
\usepackage{subfigure}
\usepackage{wrapfig}
\usepackage{cleveref}

\newcommand{\E}{\mathbb{E}}
\newcommand{\e}{\epsilon}
\newcommand{\te}{\tilde{\epsilon}}
\newcommand{\eps}{\epsilon}
\newcommand{\heps}{\hat{\epsilon}}

\newcommand{\one}{\mathbf{1}}
\newcommand{\rev}{\textsc{Rev}}
\newcommand{\opt}{\textsc{Opt}}
\newcommand{\abs}[1]{\vert #1 \vert}
\newcommand{\event}{\mathcal{E}}
\newcommand{\ps}{\mathbf{p}}
\newcommand{\vs}{\mathbf{v}}

\newtheorem{theorem}{Theorem}[section]
\newtheorem{heuristic algorithm}{Heuristic Algorithm}
\newtheorem{lemma}{Lemma}

\newtheorem{corollary}{Corollary}

\newtheorem{innercustomgeneric}{\customgenericname}
\providecommand{\customgenericname}{}
\newcommand{\newcustomtheorem}[2]{%
  \newenvironment{#1}[1]
  {%
   \renewcommand\customgenericname{#2}%
   \renewcommand\theinnercustomgeneric{##1}%
   \innercustomgeneric
  }
  {\endinnercustomgeneric}
}

\newcustomtheorem{customtheorem}{Theorem}
\newcustomtheorem{customlemma}{Lemma}

\usepackage{etoolbox} 
\newif\ifincludeAppendix

\newif\ifincludeMainbody

\begin{document}

\title{Learning to Price Against a Moving Target}
\author{
Renato Paes Leme \\ Google Research \\ {\tt renatoppl@google.com} \and 
Balasubramanian Sivan \\ Google Research \\ {\tt balusivan@google.com} \and
Yifeng Teng \\ UW-Madison \\ {\tt yifengt@cs.wisc.edu} \and 
Pratik Worah \\ Google Research \\ {\tt pworah@google.com}
}
\date{}

\maketitle
\thispagestyle{empty}

\begin{abstract}
In the Learning to Price setting, a seller posts prices over time with the goal of maximizing revenue while learning the buyer's valuation. This problem is very well understood when values are stationary (fixed or iid). Here we study the problem where the buyer's value is a moving target, i.e., they change over time either by a stochastic process or adversarially with bounded variation. In either case, we provide matching upper and lower bounds on the optimal revenue loss. Since the target is moving, any information learned soon becomes out-dated, which forces the algorithms to keep switching between exploring and exploiting phases.
\end{abstract}


\section{Introduction}

Inspired by applications in electronic commerce, we study a problem where a seller repeatedly interacts with a buyer by setting prices for an item and observing whether the buyer purchases or not. These problems are characterized by two salient features: (i) binary feedback: we only observe if the buyer purchased or not, at the price we posted; (ii) discontinuous loss function: pricing just below the buyer's valuation incurs a small loss while pricing just above it incurs a large loss since it results in a no-sale.

This problem has been studied with many different assumptions on how the buyer valuation $v_t$ changes over time: fixed over time and i.i.d. draws each round were studied in \citep{kleinberg2003value, DPS19,cesa2019dynamic}, deterministic contextual \citep{amin2014repeated,cohen2016feature,lobel2016multidimensional,Intrinsic18,liu21}, contextual with parametric noise \cite{nazerzadeh2016} and contextual with non-parametric noise \cite{virag09,krishnamurthy2020contextual}. All those models are stationary in the sense that the buyer's model is  i.i.d. across time. The exceptions to this are algorithms that consider valuations that are drawn  adversarially  \citep{kleinberg2003value}, but that work still compares with the best single price in hindsight. I.e., even though the buyer model is non-stationary, the benchmark still is.

Our main goal in this paper is to explore settings where both the buyer model and the benchmark are non-stationary. We will compare our revenue with the first-best benchmark, namely, the sum of the buyer's value at every single step. We will however assume that the buyer's valuation moves slowly.

\paragraph{Motivation} Our main motivation for this study is online advertising. Display ads are mostly sold through first price auctions with reserve prices \cite{paes2020competitive}. In many market segments, the auctions are thin, i.e., there is just one buyer, who bids just above the reserve when his value exceeds the reserve (to both pay as little as possible, and not reveal their true value) and doesn't bid otherwise. This scenario effectively offers just binary feedback, and also makes reserve price the only pricing tool (i.e., not much auction competition). 
To see why buyer value changes are typically slow, and unknown to the seller: the effective value of a buyer, even for two identical queries, is similar but not exactly the same due to factors such as remaining budget. A common scenario is that a  buyer has a spend target stating a target $\theta_t$ of their daily budget to be spent by time $t$. Bids often become a function of the ratio between the actual spend and the target spend. The auction platform doesn't know the targets/bidding formula, but it can use the fact that both target and actual spend, and hence the bids, will change smoothly over time.

Another important motivation is to effectively price buyers who are learning about their own valuation. This is a common setup in finance \cite{shreve2004stochastic} where traders constantly acquire new information about the products they are trading, and update their valuations accordingly.

Our results and techniques are presented in Section \ref{sec:results} after we formally define our model in Section \ref{sec:setting}.

\paragraph{Related Work} 
Our work is situated in the intersection of two lines of work in online learning: online learning for pricing (discussed earlier in the introduction) and online learning with stronger benchmarks, such as tracking regret \cite{herbster2001tracking,luo2015achieving}, adaptive regret \cite{hazan2007adaptive}, strongly adaptive online learning \cite{daniely2015strongly} and shifting bandits \cite{foster2016learning,lykouris2018small}. The difficulty in applying this line of work to pricing problems is that even when the valuation $v_t$ changes slightly, the loss function itself will change dramatically for certain prices. Instead here, we exploit the special structure to the revenue loss to obtain better regret bounds.

There is another line of work that studies revenue maximization in the presence of evolving buyer values \citep{KLN13, PST14, CDKS16}. While all these works consider the cumulative value over time as benchmark, there are important differences, the first two papers have full feedback: they design mechanisms that solicit buyer bids.~\citeauthor{CDKS16} shoot for simple pricing schemes yielding constant factor approximations, while we seek to obtain much closer to the optimal. Moreover, in their model the values evolve only when the buyer purchases the good.

\section{Setting}\label{sec:setting}

\paragraph{General setting.} A seller repeatedly interacts with a buyer over $T$ time steps, offering to sell an item at each step. 
The buyer's value is $v_t\in[0,1]$ at time step $t\in[T]$, and is changing across time. The seller has no knowledge of the buyer's value at any time, not even at $t=1$. At each time step $t$, the seller posts a price $p_t\in[0,1]$ to the buyer, and the buyer purchases the item if and only if she can afford to buy it, i.e. $p_t\leq v_t$. The binary signal
\[\sigma_t=\one \{ v_t \geq p_t \} \in \{0,1\}\]
of whether the item is sold or not is the only feedback the seller obtains at each time step. The seller's objective is to maximize his the total revenue, i.e.
$\rev = \sum_{t=1}^{T} p_t \sigma_t.$

\paragraph{Loss metrics.}
The benchmark we are comparing to is the \textit{hindsight optimal} revenue we can get from the buyer, which is the sum of her value across all time periods:
$\opt = \sum_{t=1}^{T} v_t.$
The \textit{revenue loss} at any time step $t$ is defined by
\(\ell_{\textsc{R}}(p_t,v_t)=v_t-p_t\sigma_t\),
and the revenue loss of any price vector $\ps=(p_1,\cdots,p_T)$ for value profile $\vs=(v_1,\cdots,v_T)$ is defined by 
\[\ell_\textsc{R}(\ps,\vs)=\frac{1}{T}(\opt-\rev)=\frac{1}{T}\sum_{t=1}^{T} (v_t-p_t\sigma_t).\]
The \textit{symmetric loss} at any time step $t$ is defined by $\ell_1(p_t,v_t)=|v_t-p_t|$, and the total symmetric loss of any price vector $\ps=(p_1,\cdots,p_T)$ for value profile $\vs=(v_1,\cdots,v_T)$ is defined by 
\[\ell_1(\ps,\vs)=\frac{1}{T}\sum_{t=1}^{T} |v_t-p_t|.\]
Intuitively, the revenue loss determines how much revenue we lose compared to a hindsight optimal selling strategy. The symmetric loss determines how close our price guesses are from the correct value vector of the buyer.

\section{Our Results and Techniques}\label{sec:results}

Next we describe our main results in this model. Our results will be given in terms of the changing rate, which we define as follows: consider a sequence of buyer valuations $v_1, v_2, \hdots, v_T$ and a sequence $\epsilon_1,\dots,\epsilon_{T-1}$ (all $\epsilon_t \leq 1$).
We say that a sequence of buyer valuations $\{v_t\}_{t=1..T}$ has changing rate $\{\epsilon_t\}_{t=1..T-1}$ whenever
\begin{equation}\label{val:eqn}
    \abs{v_{t+1} - v_t} \leq \epsilon_t.
\end{equation}
The average changing rate is $\bar \epsilon = \frac{1}{T-1}\sum_{t=1}^{T-1} \eps_t$. Our guarantees are a function of the average changing rate $\bar \epsilon$ but our algorithms are \emph{agnostic} to both the changing rate sequence and its average (except in warmup section~\ref{sec:known_constant_eps}). 

We will consider the problem in two environments:
\begin{enumerate}
\item Adversarial: an adaptive adversary chooses $v_t$'s.
\item Stochastic: an adaptive adversary chooses a mean-zero distribution supported in $[-\epsilon_t, \epsilon_t]$ from which $v_{t+1}-v_t$ is drawn so that $v_t$ is a bounded martingale. 
More generally, our results hold for any stochastic process satisfying Azuma's inequality.
\end{enumerate}

We analyze three settings in total: symmetric loss in the adversarial environment, and revenue loss in both the adversarial and stochastic environments. In each of the three settings, we design our algorithms gradually starting from the simple case where the changing rate is fixed and known (warmup, Section~\ref{sec:known_constant_eps}), then fixed and unknown (Section~\ref{sec:unknown_constant_eps})), and finally dynamic and unknown (Section~\ref{sec:unknown_dynamic_eps})).



\paragraph{Dynamic, non-increasing changing rates} In Section~\ref{sec:unknown_dynamic_eps}, where we study dynamic and unknown $\epsilon_t$, we focus primarily on the case where the changing rates $\epsilon_t$ are non-increasing over time. This is motivated by situations where buyers use a learning algorithm with non-increasing learning rates (quite common) to determine their value, or using a controller to bid that stabilizes once the value approaches the ground truth. For symmetric loss alone (Theorem \ref{thm:adv-unknownchangingeps-symmetric}), we give guarantees for any sequence $\epsilon_t$  (i.e., not just non-increasing) but we get an additional $\log(T)$-factor in loss. 

\paragraph{Our results} For symmetric loss, we develop an algorithm with average loss $\tilde{O}(\bar\epsilon)$ (Theorem~\ref{thm:adv-decreasingeps-symmetric}) [and a slightly larger loss of $\tilde{O}(\bar \epsilon \log T)$ when the $\epsilon_t$'s are not necessarily non-increasing (Theorem~\ref{thm:adv-unknownchangingeps-symmetric})]. For revenue loss we show average loss of $\tilde{O}(\bar \epsilon^{1/2})$ in the adversarial setting (Theorem \ref{thm:adv-decreasingeps-pricing}) and $\tilde{O}(\bar \epsilon^{2/3})$ in the stochastic setting (Theorem \ref{thm:rand-decreasingeps-pricing}). Throughout, we use $\tilde{O}(f)$ to denote $O(f)\text{polylog}(1/\bar\epsilon) + o(1)$ where the $o(1)$ is with respect to $T$. Surprisingly, our loss bounds for none of the three settings in the general case of dynamic and unknown changing rates (Section~\ref{sec:unknown_dynamic_eps}) can be improved even if the changing rate is fixed and known (Section~\ref{sec:known_constant_eps}, Theorems~\ref{thm:adv-knownfixedeps-symmetric},~\ref{thm:adv-knownfixedeps-pricing},~\ref{thm:rand-knownfixedeps-pricing}). I.e., in our fairly general model, knowing the rate of change of valuation, or having the rate of change be fixed, don't provide much help in learning the valuation function to profitably price against it! 

\paragraph{Our techniques}
\emph{Step 1: fixed and known $\epsilon$ (Section \ref{sec:known_constant_eps}):} Here, the algorithm keeps in each period a confidence interval $[\ell_t, r_t]$ containing $v_t$. If $r_t - \ell_t$ is small, we price at the lower endpoint $\ell_t$, resulting in a larger interval $[\ell_{t+1}, r_{t+1}] = [\ell_t-\epsilon, r_t+\epsilon]$ in the next iteration. Once the interval is large enough, we decrease its size by a binary search to decrease and re-start the process. The algorithm thus keeps alternating between exploring and exploiting, where the length of the interval decides when we do what.

\emph{Step 2: fixed and unknown $\epsilon$ (Section \ref{sec:unknown_constant_eps}):} Here, we start guessing a very small value of $\epsilon$, say $\hat \epsilon = 1/T$ and we behave as if we knew $\epsilon$. If we detect any inconsistency with our assumption, we double our guess: $\hat \epsilon \leftarrow 2 \hat \epsilon$. It is easy to detect when the value is below our confidence interval (since we always price at the lower point) but not so easy to detect when it is above. To address this issue, we randomly select certain rounds (with probability proportional to a power of our estimate $\hat \epsilon$) to be \emph{checking rounds}. In those rounds, we price at the higher end of the interval.

\emph{Step 3: dynamic, non-increasing and unknown $\epsilon_t$ (Section \ref{sec:unknown_dynamic_eps}):} We again keep an estimate $\hat \epsilon$ like in Step 2, but now we adjust it in both directions. We increase $\hat \epsilon$ as in the previous step if we observe anything inconsistent with the guess. For the other direction, we \emph{optimistically} halve the guess ($\hat \epsilon \leftarrow \hat \epsilon/2$) after we spend $1/\hat \epsilon$ periods with the same guess.

\paragraph{Comment on average versus total loss.}  Our results are all framed in terms of the average loss instead of the total loss, thus suppressing the factor $T$. To see why, note that even if we knew $v_{t-1}$ exactly for each $t$, and even if values evolved stochastically, we would already incur an $O(\epsilon)$ loss per step leading to a total regret of $O(T\epsilon)$. Consequently, the main question is not the dependence on $T$ (which is necessarily linear and cannot become any worse), but how much worse does the dependence on $\epsilon$ get, per time step. To see this in action, it is instructive to compare our algorithms with the algorithm in \cite{kleinberg2003value} which has sublinear loss with respect to a \emph{fixed price benchmark}. For a fixed $\epsilon$ consider the periodic sequence $v_t$ that starts at zero and increases by $\epsilon$ in each period reaching $1$ and then decreases by $\epsilon$ in each period reaching $0$ and starts climbing again. The first-best benchmark is $\sum_t v_t = \frac{T}{2} + O(\epsilon)$ while the best fixed price benchmark is $\frac{T}{4} + O(\epsilon)$. Our algorithm guarantees total revenue $T(\frac{1}{2} -O(\sqrt{\epsilon}))$, while~\citeauthor{kleinberg2003value} guarantee a revenue of $T \frac{1}{4} - O(\sqrt{T})$. In the supplementary material we show that for this example their algorithm indeed suffers a total loss of $\Omega(T)$ with respect to the first-best benchmark, while our algorithms suffer $T\tilde{O}(\sqrt{\epsilon})$, i.e., much better dependence on $\epsilon$.

\section{Warmup: Buyer's Changing Rate is Fixed, Known}\label{sec:known_constant_eps}

We begin by studying the symmetric loss in the adversarial environment. The result is straightforward via a binary search algorithm that keeps track of a \textit{\textbf{confidence interval}} $[\ell_t,r_t]$ that contains the true value $v_t$ in each time step. 

\begin{theorem}\label{thm:adv-knownfixedeps-symmetric}
If the buyer has adversarial values and a fixed changing rate $\epsilon$ that is known to the seller, there is an online pricing algorithm that achieves symmetric loss $O(\e)$. Further, no online algorithm can obtain symmetric loss less than $\Omega(\e)$.  
\end{theorem}

\begin{proof}
Firstly, consider the value sequence to be an unbiased random walk with step size $\eps$, and starting point $v_1=\frac{1}{2}$. The symmetric loss from each time step $t$ is $\Omega(\eps)$ even if $v_{t-1}$ is known to the seller, since the seller does not know whether $v_{t}=v_{t-1}+\eps$ or $v_{t}=v_{t-1}-\eps$. Thus no online algorithm can obtain a symmetric loss $o(\epsilon)$ for this instance.

Then we show that a naive binary search algorithm achieves symmetric loss $O(\epsilon)$. The algorithm keeps track of a \textit{\textbf{confidence interval}} $[\ell_t,r_t]$ that contains the true value $v_t$ in each time step. 

\begin{algorithm}[htb]
\caption{Symmetric loss minimizing algorithm for adversarial buyer with known changing rate $\eps$}
\label{alg:adv-knownfixedeps-symmetric}
\begin{algorithmic}
 \STATE $\ell_1\leftarrow 0$
 \STATE $r_1\leftarrow 1$
 \FOR{each time step $t$}{
 \STATE Price at $p_t\leftarrow\frac{r_t+\ell_t}{2}$
 \IF{the current value $v_t<p_t$}{
    \STATE $\ell_{t+1}\leftarrow\max(0,\ell_t-\eps)$
    \STATE $r_{t+1}\leftarrow \min(1,p_t+\eps)$
 }
 \ELSE{
    \STATE $\ell_{t+1}\leftarrow \max(0,p_t-\eps)$
    \STATE $r_{t+1}\leftarrow \min(1,r_{t}+\eps)$
 }
 \ENDIF
 }
 \ENDFOR
\end{algorithmic}
\end{algorithm}
At time $t$ we have a confidence interval $[\ell_t,r_t]\ni v_t$ that contains the true value of the buyer, since the buyer's value only changes by at most $\eps$. If the buyer can afford to pay $p_t=\frac{\ell_t+r_t}{2}$, i.e. $v_{t}\geq p_t$, it means that $v_{t+1}\in[p_t-\eps,r_t+\eps]$; otherwise, it means that $v_{t+1}\in[\ell_t-\eps,p_t+\eps]$. In both cases, we almost halve the length of the confidence interval. Let $a_t$ be the length of the confidence interval at time $t$. Since the price is always the middle point of the interval, the symmetric loss of each time step is at most $\frac{a_t}{2}$. As $a_t$ starts with 1 and almost halves in each step until $a_t=O(\eps)$, we have $\frac{1}{2}\sum_{t}a_t=O(T\eps)$, which upper bounds the total symmetric loss of all time steps in the algorithm.

\end{proof}

Next, we extend the algorithm for symmetric loss to an optimal algorithm for revenue loss.

\begin{theorem}\label{thm:adv-knownfixedeps-pricing}
If the buyer has adversarial values and a fixed changing rate $\epsilon$ that is known to the seller, there exists an online pricing algorithm with revenue loss $\tilde{O}(\e^{1/2})$. Further, no online algorithm can obtain a revenue loss less than $\Omega(\e^{1/2})$. 
\end{theorem}

\begin{proof}
We assume that $\eps\geq\frac{1}{T}$. Otherwise, we can run the algorithm with $\eps'=\frac{1}{T}$ and get revenue loss $\tilde{O}((\eps')^{1/2})=o(1)=\tilde{O}(\eps^{1/2})$.

The idea of the algorithm is as follows. The sale process starts with the seller not knowing the initial value $v_1$ of the buyer. However, since the buyer's value only changes by at most $\eps$ in each step, the seller can quickly locate the buyer's value within $O(\eps)$ error by binary search.
Such a small confidence interval that contains the buyer's current value extends by $\eps$ in both upper bound and lower bound after each step. We propose an algorithm that repeatedly prices at the bottom of the confidence interval and re-locates the buyer's current value whenever the confidence interval becomes too wide. 

Firstly, we have the following building-block algorithm that quickly returns the range of the buyer's value with additive accuracy $O(\eps)$. In the algorithm, timestamp $t$ increases by one after any pricing query. 

\begin{algorithm}[H]
\caption{Locating the value of a buyer with changing rate $\eps$ to an interval of length $a\leq 6\eps$}
\label{alg:binarysearch}
\begin{algorithmic}
 \STATE {\bfseries Input:} current confidence interval $[\ell_t,r_t]$
 \STATE Run Algorithm~\ref{alg:adv-knownfixedeps-symmetric} until confidence interval $[\ell_t,r_t]$ satisfies $r_t-\ell_t<4\eps$
 \STATE \textbf{Return}  {$[\ell_t-\eps,r_t+\eps]$}
\end{algorithmic}
\end{algorithm}
The algorithm repeatedly does binary search until the confidence interval has length $O(\eps)$. If at the beginning we have a confidence interval of length $b$ and finally we have a confidence interval of length $a$, thus the total number of steps needed is $O(\log\frac{b}{a})$, and the total loss of the algorithm is $O(\log\frac{b}{a})$ since the loss of each step is at most 1.

Next, we present the entire pricing algorithm for an adversarial buyer with changing rate $\eps$ using Algorithm~\ref{alg:binarysearch} as a building block.

\begin{algorithm}[H]
\caption{Revenue loss minimizing algorithm for adversarial buyer with known changing rate $\eps$}
\label{alg:adv-knownfixedeps-pricing}
\begin{algorithmic}
\FOR{each phase $[t_0+1,t_0+m]$ of length $m=\eps^{-1/2}$}{
    \STATE Apply Algorithm~\ref{alg:binarysearch} to locate the current value of the buyer in an interval $[\ell_t,r_t]$ with length $\sqrt{\eps}$
    \FOR{each step $t$ in the phase}{
        \STATE Price at $p_t\leftarrow \ell_t$
        \STATE $[\ell_{t+1},r_{t+1}]\leftarrow[\ell_{t}-\eps,r_t+\eps]$
    }
    \ENDFOR
 }
 \ENDFOR
\end{algorithmic}
\end{algorithm}

Now we analyze the revenue loss of the algorithm. In each phase of the algorithm with $m$ time steps, we first initialize and locate the value of the buyer to a small range $\sqrt{\eps}$. The revenue loss incurred by Algorithm~\ref{alg:binarysearch} is $O(\log\frac{1}{\eps})={\tilde O(1)}$ in each phase (actually $O(1)$ after the first phase since the length of confidence interval only needs to shrink by constant fraction). Then we price at the bottom of the confidence interval for $O(\sqrt{\eps})$ steps for $m=\sqrt{\eps}$ steps. Since the confidence interval $[\ell_t,r_t]$ expands by $2\eps$ each time, it is equivalent to say we wait until the confidence interval has length $3\sqrt{\eps}$. The revenue loss from each step is $\leq 3\sqrt{\eps}$ since the item is bought by the buyer every time, thus $O(\sqrt{\eps})\cdot 3\sqrt{\eps}=O(\eps)$ in the entire phase; then we use binary search to narrow the confidence interval by $\frac{1}{2}$, which has revenue loss $O(1)$ since it takes only $O(1)$ steps in Algorithm~\ref{alg:binarysearch}. Each phase takes $\Theta(\sqrt{\eps})$ steps with loss $O(\eps)$, therefore there are $O(T\eps^{-1/2})$ phases in total with loss $O(T\eps^{-1/2})\cdot O(\eps)=O(T\eps^{1/2})$. Thus the total loss of the algorithm is $\tilde{O}(T\eps^{1/2})$, with average revenue loss $\tilde{O}(\eps^{1/2})$. 

We show that such loss is tight, that no online algorithm can obtain a revenue loss $o(\eps^{1/2})$. By Yao's minimax principle, to reason that there exists some adversarial buyer such that no randomized online algorithm can get revenue loss $o(\eps^{1/2})$, it suffices to show that there exists a random adversarial buyer such that no online deterministic algorithm can get such low revenue loss. The buyer's value sequence is predetermined as follows. In the beginning, the buyer has value $v_0=\frac{1}{2}$. The entire time horizon $[T]$ is partitioned to $T\sqrt{\eps}$ phases, each with length $\eps^{-1/2}$. In each phase starting with time $t_0+1$ and ending with time $t_0+\eps^{-1/2}$, the values of the buyer form a monotone sequence: with probability $\frac{1}{2}$, $v_{t_0+j}=v_{t_0}+j\eps$, $\forall j\in [\eps^{-1/2}]$; with probability $\frac{1}{2}$, $v_{t_0+j}=v_{t_0}-j\eps$, $\forall j\in [\eps^{-1/2}]$. For any deterministic algorithm with price $p_t$ at each time, when any phase begins, the algorithm needs to decide the pricing strategy without knowing which value instance of the buyer will be realized. Let $\hat{t}\in[t_0+1,t_0+\eps^{-1/2}]$ be the first time step in the phase such that $v_{t_0}-(t-t_0)\eps<p_{t}$. 

If such $\hat{t}$ exists, then the revenue loss at time $\hat{t}$ is $v_{t_0}-(\hat{t}-t_0)\eps=v_{t_0}-O(\eps^{1/2})$ when the values of the buyer decrease in the phase, and this case happens with probability $\frac{1}{2}$. Then in expectation the revenue loss in this phase is $\Omega(v_{t_0})-O(\eps^{1/2})$. Notice that the starting value $v_{t_0}$ of each phase form a unbiased random walk sequence with step size $\eps^{1/2}$, since the buyer starts with $v_0=\frac{1}{2}$, therefore with constant probability $v_{t_0}\geq\frac{1}{2}$. Thus we can also claim that the expected revenue loss in the phase is $\Omega(1)$. 
If such $\hat{t}$ does not exist, it means that the algorithm has identical information from both instances of the buyer in the entire phase, since the buyer can always afford to purchase the item. As $p_{t}\leq v_{t_0}-(t-t_0)\eps$ for each time $t$, the revenue loss of the algorithm when the values of the buyer increase in the phase is at least 
    $\sum_{t_0<t\leq t+\eps^{-1/2}}(v_{t_0}+(t-t_0)\eps-p_t)
    \geq\sum_{t_0<t\leq t+\eps^{-1/2}}2(t-t_0)\eps=\Omega(1)$.

Therefore in both cases, the revenue loss in the phase is $\Omega(1)$ in this phase for any deterministic algorithm. As there are $T\eps^{1/2}$ phases, the total revenue loss from all phases is at least $\Omega(T\eps^{1/2})$. Thus no randomized algorithm can get average revenue loss $o(\eps^{1/2})$.

\end{proof}

If the buyer's value evolves stochastically across time, the stochasticity helps us incur less revenue loss. To be more specific, when the buyer's value is stochastic and forms a martingale, after every $\eps^{-2/3}$ steps, although the buyer's value can change by as large as $\eps^{-2/3}\cdot\eps=\eps^{1/3}$, with high probability the buyer's value only changes by $\eps^{2/3}$. Thus compared to Algorithm~\ref{alg:adv-knownfixedeps-pricing}, we can extend each phase's length while maintaining a shorter confidence interval. We state the algorithm and the revenue loss below.

\begin{algorithm}[H]
\caption{Revenue loss minimizing algorithm for stochastic buyer with known changing rate $\eps$}
\label{alg:rand-knownfixedeps-pricing}
\begin{algorithmic}
 \STATE $[\ell,r]\leftarrow [0,1]$
 \FOR{each phase $[t_0+1,t_0+m]$ of length $m=\eps^{-2/3}$}{
    \STATE Apply Algorithm~\ref{alg:binarysearch} to narrow down the confidence interval $[\ell,r]$ to length $6\eps$
    \FOR{each step $t$ in the phase}{
         \STATE Price at $p_t\leftarrow\ell-4\eps^{2/3}\sqrt{\log \frac{1}{\eps}}$
    }
    \ENDFOR
 }
 \ENDFOR
\end{algorithmic}
\end{algorithm}

\begin{theorem}\label{thm:rand-knownfixedeps-pricing}
If the buyer has stochastic value and a fixed changing rate $\epsilon$ that is known to the seller, Algorithm~\ref{alg:rand-knownfixedeps-pricing} has revenue loss $\tilde{O}(\e^{2/3})$. Furthermore, no online algorithm can obtain a revenue loss less than $\Omega(\e^{2/3})$.
\end{theorem}

\begin{proof}
We assume that $\eps\geq\frac{1}{T}$. Otherwise, we can run the algorithm with $\eps'=\frac{1}{T}$ and get revenue loss $\tilde{(\eps')^{2/3}}=o(1)=\tilde{O}(\eps^{2/3})$.

When the buyer's value is stochastic and forms a martingale, after every $\eps^{-2/3}$ steps, although the buyer's value can change by as large as $\eps^{-2/3}\cdot\eps=\eps^{1/3}$, with high probability the buyer's value only changes by $\eps^{2/3}$. Thus compared to Algorithm~\ref{alg:adv-knownfixedeps-pricing}, we can extend the length of each phase, while maintaining shorter confidence interval. The detailed algorithm is as follows. 

\begin{algorithm}[htb]
\caption{Revenue loss minimizing algorithm for stochastic buyer with known changing rate $\eps$}
\begin{algorithmic}
 \STATE $\ell\leftarrow 0$
 \STATE $r\leftarrow 1$
 \FOR{each phase $[t_0+1,t_0+m]$ of length $m=\eps^{-2/3}$}{
    \STATE Apply Algorithm~\ref{alg:binarysearch} to narrow down the confidence interval $[\ell,r]$ to length $6\eps$
    \FOR{each step $t$ in the phase}{
         \STATE Price at $p_t\leftarrow\ell-4\eps^{2/3}\sqrt{\log \frac{1}{\eps}}$
    }
    \ENDFOR
 }
 \ENDFOR
\end{algorithmic}
\end{algorithm}

Now we analyze the revenue loss of the algorithm. In each phase $[t_0+1,t_0+m]$ with $m=\eps^{-2/3}$, firstly the binary search algorithm is called with a revenue loss of $O(\log\frac{1}{\eps})=\tilde{O}(1)$. Suppose that at time $t_1$ the search algorithm ends, and we get an estimate of $v_{t_1}$ with additive accuracy $6\eps$. Let $\delta=4\eps^{2/3}\sqrt{\log \frac{1}{\eps}}$. In the remaining of the phase, the price is fixed to be $\ell-\delta$. For any $t\in[t_1,t_0+m]$, by Azuma's inequality, the probability that $v_{t}\not\in[v_{t_1}-\delta,v_{t_1}+\delta]$ is
\begin{equation}\label{eqn:Azmua}
\Pr[v_{t}\not\in[v_{t_1}-\delta,v_{t_1}+\delta]]\leq 2\exp\left(-\frac{\delta^2}{2m\eps^2}\right)<\frac{1}{m^2}.\end{equation}
By union bound, the probability that exists $t\in[t_1,t_0+m]$ such that $v_{t}\not\in[v_{t_1}-\delta,v_{t_1}+\delta]$ is at most $\frac{1}{m}$; in this case, the revenue loss in the phase is at most 1 per step thus $m$ in total, so the expected revenue loss contributed from the case is $O(1)$. If $v_{t}\in[v_{t_1}-\delta,v_{t_1}+\delta]$ for every $t\in[t_1,t_0+m]$, the revenue loss for setting price $\ell-\delta\geq v_{t}-6\eps-\delta$ is at most $O(\eps+\delta)$ for each step, thus $mO(\eps+\delta)=O(\sqrt{\log\frac{1}{\eps}})=\tilde{O}(1)$ in total. Combine both cases above, we know that the expected revenue loss in each phase is $\tilde{O}(1)$, thus $\tilde{O}(T\eps^{2/3})$ in total and $\tilde{O}(\eps^{2/3})$ revenue loss on average since there are $\frac{T}{m}=T\eps^{2/3}$ phases.

Now we show that such loss is tight, even when the buyer's values form an unbiased random walk with step size $\eps$ and starting value $\frac{1}{2}$.

Consider a game with multiple phases and the following more powerful seller. Each phase has at most $2m=2\e^{-2/3}$ time steps. At the beginning of each phase, the seller is told the exact value $v_0$ of the buyer at the current time. Then at each time step $t\in[1,2m]$, the current value $v_t$ increases or decreases by $\e$ with equal probability, and the seller sets a price $p_t$. Whenever $p_t>v_t$, the phase terminates and incurs a loss $\ell=v_t$. If $p_t<v_t$ for the entire phase, the loss $\ell$ of the phase is defined by $\sum_{t\leq 2m}\ell_t=\sum_{t\leq 2m}(v_t-p_t)$. Then the game goes to the next phase, with $v_0$ reset to $v_{2m}\pm\e$. The entire game ends when the total number of time steps in each phase sum up to $T$.

We notice that in the game, the seller is strictly more powerful, while the loss does not increase. In each phase before the phase terminates, the seller learns no information from the pricing. Thus without loss of generality, assume that given initial value $v_0$ of the buyer, the seller proposes a sequence of deterministic prices (randomness won't help) $p_1,p_2,\cdots,p_{2m}$, and then the values $v_1,\cdots,v_{2m}$ are realized randomly.

We do not count the loss contributed from the first $m$ steps. Since Chernoff bound is tight up to constant factor for $U\{-1,1\}$ random variables, i.e. $\Pr[v_t<v_0-\e\sqrt{t}]\geq\Omega(1)$ for $t\geq m$. Thus for any $t\geq m$ if $p_t\geq v_0-\e\sqrt{m}$, with constant probability the phase stops at time $t$ with loss $v_t$, which is $\Omega(1)$ in expectation. If $p_t<v_0-\e\sqrt{m}$, $\E[v_t-p_t|v_t\geq p_t]=\Omega(\e\sqrt{m})=\Omega(\e^{2/3})$ since with probability $\frac{1}{2}$, $v_t\geq v_0$. Thus the loss of the phase is at least $m\Omega(\e^{2/3})=\Omega(1)$ in expectation if the phase ends at step $2m$.

The above analysis shows that at each phase, the expected loss is $\Omega(1)$. Since there are at least $\frac{T}{m}=T\e^{2/3}$ phases, the total loss is at least $\Omega(T\e^{2/3})$ in expectation. This finishes the proof of no randomized algorithm can get revenue loss $o(\eps^{2/3})$.

\end{proof}

\section{Buyer's Changing Rate is Fixed, Unknown}\label{sec:unknown_constant_eps}
In this section, we consider the case where the buyer's value changes by fixed rate $\eps_t=\epsilon$ that is unknown to the seller. As a warm-up, we first study the symmetric loss obtained by online prices. Such a problem lets us understand better how to deal with the unknown rate of value change, and the pricing algorithm can be extended to the case of revenue loss.

\begin{theorem}\label{thm:adv-unknownfixedeps-symmetric}
If the buyer has adversarial values and a fixed changing rate $\epsilon$ unknown to the seller, there exists an online pricing algorithm with a symmetric loss $\tilde{O}(\epsilon)$. Further, no online algorithm can obtain a symmetric loss less than $\Omega(\epsilon)$.
\end{theorem}

\begin{proof}
The $\Omega(\eps)$ tightness result was shown in Theorem~\ref{thm:adv-knownfixedeps-symmetric}.
Here we propose a pricing algorithm with a symmetric loss $\tilde{O}(\epsilon)$.
Compared to the case where $\eps$ is known to the seller, the seller can use a guess $\heps$ to replace the true rate $\eps$, and run the algorithms in the known $\eps$ setting. The algorithm starts with $\heps=\frac{1}{T}$, and doubles $\heps$ whenever the algorithm detects that the changing rate of the value exceeds $\heps$.

\begin{algorithm}[htb]
\caption{Symmetric loss minimizing algorithm for adversarial buyer with unknown changing rate $\eps$}
\label{alg:adv-unknownfixedeps-symmetric}
\begin{algorithmic}
 \STATE $[\ell_1,r_1]\leftarrow [0,1]$
 \STATE $\heps\leftarrow\frac{1}{T}$
 \WHILE{$\heps<\frac{1}{2}$}{
    \FOR{each three consecutive time steps $t,t+1,t+2$}{
        \STATE Set price $p_t\leftarrow \ell_t$
        \STATE Set price $p_{t+1}\leftarrow r_t+\heps$
        \STATE Set price $p_{t+2}\leftarrow \frac{\ell_t+r_t}{2}$
        \IF{the seller finds $v_t<\ell$ or $v_{t+1}\geq r_t+\heps$}{
            \STATE $\heps\leftarrow\heps\cdot2$
            \STATE break
        }
        \ENDIF
        \IF{$v_{t+2}<p_{t+2}$}{
            \STATE $\ell_{t+3}\leftarrow \ell_t-3\heps$, $r_{t+3}\leftarrow p_{t+2}+\heps$
        }
        \ELSE{
            \STATE $\ell_{t+3}\leftarrow p_{t+2}-\heps$, $r_{t+3}\leftarrow r+3\heps$
        }
        \ENDIF
    }
    \ENDFOR
 }
 \ENDWHILE
\end{algorithmic}
\end{algorithm}

The algorithm deal with three time steps $t,t+1,t+2$ at each time, and always maintain a confidence interval $[\ell_t,r_t]$ of $v_t$ at the beginning of time step $t$, such that $\heps>\eps$, $v_{t}\in[\ell_t,r_t]$ always holds. 
Time steps $t$ and $t+1$ are ``checking steps'', to check whether the current three steps $v_{t}$ , $v_{t+1}$ and $v_{t+2}$ incur too much loss. In particular, if $\heps>\eps$, then $v_{t}\geq p_{t}$ and $v_{t+1}<p_{t+1}$ will always hold. On the other hand, if $v_{t}\geq p_{t}$, then $v_{t+1}\geq p_{t}-\eps$, $v_{t+2}\geq p_{t}-2\eps$ since the value of the buyer only changes at most $\eps$ additively, although the seller does not know the value of $\eps$. Thus $v_{t},v_{t+1},v_{t+2}\geq \ell_t-2\eps$. Similarly, if $v_{t+1}<p_{t+1}=r_t+\heps$, then $v_{t}<p_{t+1}+\eps$, $v_{t+2}<p_{t+1}+\eps$. Thus $v_{t},v_{t+1},v_{t+2}<r_t+\eps+\heps$. Since $p_{t},p_{t+1},p_{t+2}$ are all in range $[\ell_t-2\eps,r_t+\eps+\heps]$, we have for any $j\in\{t,t+1,t+2\}$, $|v_{j}-p_{j}|=O(r_t-\ell_t+\eps+\heps)$. Thus in a round of three consecutive time steps $t,t+1,t+2$, if $v_t\geq p_t$ and $v_{t+1}<p_{t+1}$, the symmetric loss from this round of three time steps is proportional to the length of confidence bound plus the true changing rate and the estimated changing rate, i.e. $O(r_t-\ell_t+\eps+\heps)$. 

To summarize, if in any three consecutive time steps the seller finds out $p_t>v_t$ or $p_{t+1}\leq v_{t+1}$, then $\heps<\eps$ and $\heps$ is doubled, and the symmetric loss obtained from time step $t,t+1,t+2$ is at most 1 each; if in any consecutive three time steps the seller observes $p_t\leq v_t$ and $p_{t+1}>v_{t+1}$, then the symmetric loss obtained from time step $t,t+1,t+2$ is $O(r_{t}-\ell_{t}+\eps+\heps)$ on average. 
Since $\heps$ can only be doubled to at most $2\eps$, we know that there are $O(\log T)$ iterations with fixed $\heps$. In each iteration, the time when $\heps$ is doubled, the symmetric loss obtained from the previous three time steps is $O(1)$; in other cases, the symmetric loss obtained from the previous three time steps is constant times the length of the confidence interval plus $O(\eps+\heps)$. Since in each phase of fixed $\heps$ the length of confidence interval starts with 1, and then decreases geometrically to $O(\heps)$, thus if the iteration contains $a$ time steps, the total length of the confidence intervals is $O(1+\heps a)=O(1+\eps a)$, and the total symmetric loss of the phase is also $O(1+\eps a)$. Thus the total symmetric loss of the algorithm is $O(\log T+T\eps)=\tilde{O}(T\epsilon)$, which is $\tilde{O}(\epsilon)$ on average.

\end{proof}

Using a similar technique of checking steps as in the case of symmetric loss, we can obtain the same revenue loss as in the setting where the seller knows the changing rate, no matter whether the buyer has adversarial or stochastic value.
\begin{theorem}\label{thm:adv-unknownfixedeps-pricing}
If the buyer has adversarial values and a fixed changing rate $\epsilon$ unknown to the seller, there exists an online pricing algorithm with revenue loss $\tilde{O}(\e^{1/2})$. Further, no online algorithm can obtain a revenue loss less than $\Omega(\e^{1/2})$.
\end{theorem}

\begin{proof}
The tightness $\Omega(\eps^{1/2})$ result is shown in Theorem~\ref{thm:adv-knownfixedeps-pricing}, so we only need to construct an algorithm with revenue loss $\tilde{O}(\e^{1/2})$.
We want to apply the same technique of ``checking steps'' as in the pricing algorithm for symmetric loss. Recall that the checking steps in Algorithm~\ref{alg:adv-unknownfixedeps-symmetric} repeatedly price at the upper bound and the lower bound of the current confidence interval of the value of the buyer to make sure the buyer's value is not too far away from the confidence bound. However, when minimizing revenue loss, the seller cannot afford to frequently check the upper bound of the confidence interval, since each time the buyer is very likely to reject the item and incurs a huge revenue loss. 
The solution to such a problem is that we can add one checking step to each phase of Algorithm~\ref{alg:adv-knownfixedeps-pricing}. By pricing at the lower bound $\ell_t$ of the confidence interval at time $t$ and the upper bound $r_{t+1}$ of the confidence interval at time $t+1$, the seller can know whether $v_{t-1}$ is far away from our confidence interval $[\ell_{t-1},r_{t-1}]$ for it. Thus for a phase with $m$ time steps and $k$ ``bad steps'' when the buyer's value is far away from the confidence interval, a random checking step can detect a bad step with probability $\frac{k}{m}$. This means that after $O(m)$ bad steps in expectation, a bad step will be detected and the algorithm will move to the next iteration with $\heps$ doubled. The algorithm is stated as follows.

\begin{algorithm}[h]
\caption{Revenue loss minimizing algorithm for adversarial buyer with unknown changing rate $\eps$}
\label{alg:adv-unknownfixedeps-pricing}
\begin{algorithmic}
\STATE $\heps\leftarrow \frac{1}{T}$
\WHILE{$\heps<\frac{1}{2}$}{
 \FOR{each phase $[t_0+1,t_0+m_{\heps}]$ of length $m_{\heps}=\heps^{-1/2}$}{
    \STATE Apply Algorithm~\ref{alg:binarysearch} to locate the current value of the buyer in an interval $[\ell_t,r_t]$ with length $\sqrt{\heps}$
    \STATE Randomly select a $t^*\in[t_0+1,t_0+m_{\heps}]$
    \FOR{each step $t$ in the phase}{
        \IF{$t\neq t^*$}{
            \STATE Price at $p_t\leftarrow \ell_t$
            \IF{$v_t<p_t$}{
                \STATE $\heps\leftarrow 2\heps$ and break (terminate the phase)
            }
            \ENDIF
        }
        \ELSE{
            \STATE Price at $p_t\leftarrow r_t$
            \IF{$v_t\geq p_t$}{
                \STATE $\heps\leftarrow 2\heps$ and break (terminate the phase)
            }
            \ENDIF
        }
        \ENDIF
        \STATE $[\ell_{t+1},r_{t+1}]\leftarrow[\ell_{t}-\eps,r_t+\eps]$
    }
    \ENDFOR
 }
 \ENDFOR
 }
 \ENDWHILE
\end{algorithmic}
\end{algorithm}
The algorithm first guesses $\heps=\frac{1}{T}$, and doubles $\heps$ whenever the algorithm detects a piece of evidence of $\heps$ being smaller than the true $\eps$. In particular, the algorithm maintains confidence bound $[\ell_t,r_t]$ that may contain the current value $v_t$ of the buyer. When $\heps$ first becomes larger than $\eps$ (thus at most $2\eps$), the algorithm will run smoothly without triggering any break statement since $v_t\in [\ell_t,r_t]$ always holds. The revenue loss is $\tilde{O}(\sqrt{\heps})=\tilde{O}(\sqrt{\eps})$ as analyzed in Theorem~\ref{thm:adv-knownfixedeps-pricing}. 

In an iteration when $\heps<\eps$, firstly there is an additional $O(\log\frac{1}{\heps})=\tilde{O}(1)$ loss for binary search initialization of the confidence interval compared to Algorithm~\ref{alg:adv-knownfixedeps-pricing} for the fixed $\eps$ setting. Let bad event $\event_t$ denote ``$v_t\geq r_t+2\eps$ or $v_t<\ell_t$''. If bad events never happen, the additional loss is at most $2\eps$ per step (thus $T\eps$ in total), since in the analysis of Algorithm~\ref{alg:adv-knownfixedeps-pricing} it has confidence interval $[\ell_t,r_t]$ rather than $[\ell_t,r_t+2\eps]$ here. 
In a phase with time steps $[t_0+1,t_0+m_{\heps}]$, if an event $\event_t$ happens, then there are two cases. If $v_t<\ell_t$, then it is detected when $p_t=\ell_t$, which almost surely happens. If $v_t>r_t+2\eps$, then it is detected when $p_{t+1}=r_{t+1}$ i.e. $t^*=t$, since $v_{t+1}\geq v_t-\eps>r_t+\eps=r_{t+1}$. 
Therefore, since $t^*$ is randomly selected from $[t_0+1,t_0+m_{\heps}]$, if $k$ events in $\event_{t_0+1},\cdots,\event_{t_0+m_{\heps}}$ happens, with probability $\frac{k}{m_{\heps}}$ a bad event gets detected. 
Thus, in an iteration with fixed $\heps<\eps$, when a bad event is detected, in expectation $m$ bad events have occurred. Each bad event will result in additional revenue loss at most 1, thus $\tilde{O}(m_{\heps})=O(\heps^{-1/2})$. The total contribution of revenue loss from bad events is at most $\sum_{i=0}^{\log T}\left(\frac{2^i}{T}\right)^{-1/2}=O(T^{1/2})=To(1)$. 
To summarize, the total revenue loss of the algorithm is $\tilde{O}(T\eps^{1/2})$ for Algorithm~\ref{alg:adv-knownfixedeps-pricing} with known $\eps$, plus the binary search cost $\tilde{O}(1)$ for locating the position of $v_t$ in each iteration of different $\heps$, plus the additional cost $O(T\eps)$ for having a slightly larger confidence interval in good events than Algorithm~\ref{alg:adv-knownfixedeps-pricing}, plus a total revenue loss of $T^{1/2}$ from the bad events. Sum up all the costs above we show that the total revenue loss of Algorithm~\ref{alg:adv-unknownfixedeps-pricing} is $T\tilde{O}(\eps^{1/2})$ for all time steps, thus $\tilde{O}(\eps^{1/2})$ on average. 

\end{proof}

When the buyer has stochastic value, we can modify Algorithm~\ref{alg:adv-unknownfixedeps-pricing} for an adversarial buyer such that in each phase is replaced by a phase in Algorithm~\ref{alg:rand-knownfixedeps-pricing}, with the normal pricing step $t\neq t^*$ pricing at $p_t=\ell_{t_0}-\tilde{O}(\heps^{-1/3})$, and each checking step $t^*$ at price $p_{t^*}=r_{t_0}+\tilde{O}(\heps^{-1/3})$. The analysis is almost identical to Theorem~\ref{thm:adv-unknownfixedeps-pricing}.

\begin{theorem}\label{thm:rand-unknownfixedeps-pricing}
If the buyer has stochastic value and a fixed changing rate $\epsilon$ unknown to the seller, there exists an online pricing algorithm with revenue loss $\tilde{O}(\e^{2/3})$. Further, no online algorithm can obtain a revenue loss less than $\Omega(\e^{2/3})$.
\end{theorem}

\begin{proof}
The $\Omega(\e^{2/3})$ tightness of the theorem follows from Theorem~\ref{thm:rand-knownfixedeps-pricing}. Now we provide an algorithm with revenue loss $\tilde{O}(\e^{2/3})$. The algorithm combines Algorithm~\ref{alg:rand-knownfixedeps-pricing} for stochastic buyer with known changing rate and Algorithm~\ref{alg:adv-unknownfixedeps-pricing} for adversarial buyer with unknown changing rate as follows. 
For every iteration with guess $\heps$ in Algorithm~\ref{alg:adv-unknownfixedeps-pricing}, replace each phase of length $\heps^{-1/2}$ by a phase in Algorithm~\ref{alg:rand-knownfixedeps-pricing} with a random checking step that prices at the top of the confidence interval which tries to find the evidence of $\heps<\eps$. The analysis is almost the same as the analysis for Algorithm~\ref{alg:adv-unknownfixedeps-pricing}.

\begin{algorithm}[htb]
\caption{Revenue loss minimizing algorithm for stochastic buyer with unknown changing rate $\eps$}
\label{alg:rand-unknownfixedeps-pricing}
\begin{algorithmic}
\STATE $\heps\leftarrow \frac{1}{T}$
\WHILE{$\heps<\frac{1}{2}$}{
 \FOR{each phase $[t_0+1,t_0+m_{\heps}]$ of length $m_{\heps}=\heps^{-2/3}$}{
    \STATE Apply Algorithm~\ref{alg:adv-knownfixedeps-pricing} to locate the current value of the buyer in an interval $[\ell,r]$ with length $\sqrt{\heps}$
    \STATE Randomly select a $t^*\in[t_0+1,t_0+m_{\heps}]$
    \FOR{each step $t$ in the phase}{
        \IF{$t\neq t^*$}{
            \STATE Price at $p_t\leftarrow \ell-4\heps^{-2/3}\sqrt{\log T}$
            \IF{$v_t<p_t$}{
                \STATE $\heps\leftarrow 2\heps$
                \STATE break (terminate the phase)
            }
            \ENDIF
        }
        \ELSE{
            \STATE Price at $p_t\leftarrow r+4\heps^{-2/3}\sqrt{\log T}$
            \IF{$v_t\geq p_t$}{
                \STATE $\heps\leftarrow 2\heps$
                \STATE break (terminate the phase)
            }
            \ENDIF
        }
        \ENDIF
    }
    \ENDFOR
 }
 \ENDFOR
 }
 \ENDWHILE
\end{algorithmic}
\end{algorithm}

Now we analyze the revenue loss of the algorithm. In each iteration with changing rate estimate $\heps$, let $\delta=4\heps^{-2/3}\sqrt{\log T}$. If $\heps\geq\eps$, then the probability that $v_t\not\in[\ell-\delta,r+\delta]$ is at most $\frac{1}{T^2}$ for each step $t$ (as analyzed in \eqref{eqn:Azmua}). By union bound, the probability that the algorithm falsely doubles $\heps$ when $\heps\geq\eps$ is at most $\frac{1}{T}$ in all time steps. Therefore the revenue loss contributed from this case is $O(1)$, since the total revenue loss when this case happens is at most $T$. 

Then we assume that $\heps$ is always at most $\eps$, and the rest of the analysis is identical to the analyze of Algorithm~\ref{alg:adv-unknownfixedeps-pricing}. In each iteration with fixed $\heps$, when a bad event $v_t\not\in[\ell-\delta,r+\delta]$ is detected, in expectation $m_{\heps}=\heps^{-2/3}$ bad events have occurred, thus the revenue loss contributed from bad events is $\heps^{-2/3}$ in an iteration with changing rate estimate $\heps$, thus $\sum_{\heps}\heps^{-2/3}=O(T^{2/3})$ in total. The revenue loss for each time step with good event $v_t\not\in[\ell-\delta,r+\delta]$ is $O(\delta)=\tilde{O}(\heps^{2/3})$, thus $\tilde{O}(T\heps^{2/3})=\tilde{O}(T\eps^{2/3})$ in total. Combining the revenue loss of steps with good events, steps with bad events, and cases where $\heps>\eps$, we have the total revenue loss of the algorithm for all time steps is 
$\tilde{O}(T\eps^{2/3})+O(T^{2/3})+\tilde{O}(1)=T\tilde{O}(\eps^{2/3})$.

Note: The current revenue loss of the algorithm is actually $O(\eps^{2/3}\log T)$. It can get improved to $O(\eps^{2/3}polylog\frac{1}{\eps})$, if the algorithm sets $\delta=4\heps^{2/3}\log^4\frac{1}{\heps}$, and doubles $\heps$ whenever the current total number of bad events detected for the current $\heps$ exceeds $\frac{1}{m^2}$ fraction of the current total number of time steps. By Azuma's inequality, we can argue that the probability that $\heps>\eps$ is falsely doubled is less than $\frac{1}{m^2t^2}$ in step $t$, which will only in expectation contribute to revenue loss $\frac{T}{m^2}=O(T\heps^{2/3})$ in total. The revenue loss from good events is $O(\delta)=\tilde{O}(\heps^{2/3})$ per step. Thus the average revenue loss is $\tilde{O}(\eps^{2/3})$ per step.
\end{proof}

\section{Buyer's Changing Rate is Dynamic, Unknown}\label{sec:unknown_dynamic_eps}
In this section, we study a more complicated setting where the buyer's value changes in a more dynamic way. In particular, $|v_{t+1}-v_t|$ are upper bounded by possibly different non-increasing $\epsilon_t$ that are unknown to the seller.

For the symmetric-loss minimization problem with an adversarial buyer, we show that when $\eps_t$ are non-increasing, the seller can still achieve a symmetric loss of $\tilde{O}(\bar{\eps})$ as in the case of fixed $\eps$. 

\begin{theorem}\label{thm:adv-decreasingeps-symmetric}
If the buyer has adversarial values and non-increasing changing rates $\eps_t$ unknown to the seller, then there exists an online pricing algorithm with symmetric loss $\tilde{O}(\bar{\eps})$. Furthermore, no online algorithm can obtain a symmetric loss less than $\Omega(\bar{\eps})$.
\end{theorem}

\begin{proof}
The tightness result $\Omega(\bar{\eps})$ was shown in Theorem~\ref{thm:adv-knownfixedeps-symmetric}. Now we propose an algorithm with symmetric loss $\tilde{O}(\bar{\eps})$.

We propose the following algorithm for the seller that repeatedly guess the current level of changing rate at each time step. The algorithm starts with guessing $\heps=\frac{1}{2}$ being an estimate of $\eps_t$, and reduce the value of the guess $\heps$ by a factor of $\frac{1}{2}$ if in $O(\log T)$ time steps the algorithm cannot find any evidence supporting $\heps<\eps_t$. Whenever the algorithm finds a piece of evidence that $\heps<\eps_t$, the algorithm repeatedly doubles $\heps$ and updates the confidence interval according to the new $\heps$, until the evidence of $\heps<\eps_t$ disappears. Such a dynamic update of $\heps$ keeps the symmetric loss bounded. Below we show the algorithm in detail.

\begin{algorithm}[htb]
\caption{Symmetric loss minimizing algorithm for adversarial buyer with unknown decreasing changing rate $\eps_t$}
\label{alg:adv-decreasingeps-symmetric}
\begin{algorithmic}
 \STATE Set confidence interval $[\ell_1,r_1]\leftarrow [0,1]$
 \STATE Set changing rate estimate $\heps\leftarrow\frac{1}{2}$
 \STATE Let $a\leftarrow 0$ be the recent round without bad evidence
 \STATE Let $s\leftarrow 0$ be the number of rounds for the current $\heps$
    \FOR{each round $k$ with three consecutive time steps $3k+1,3k+2,3k+3$}{
        \STATE $s\leftarrow s+1$
        \STATE Set price $p_{3k+1}\leftarrow \ell_{3k+1}$
        \STATE Set price $p_{3k+2}\leftarrow r_{3k+2}=r_{3k+1}+\heps$
        \STATE Set price $p_{3k+3}\leftarrow \frac{\ell_{3k+1}+r_{3k+1}}{2}$
        \IF{the seller finds $v_{3k+1}\geq\ell_{3k+1}$ and $v_{3k+2}<r_{3k+2}$}{
            \STATE $a\leftarrow k$
            \IF{$v_{3k+3}<p_{3k+3}$}{
                \STATE $[\ell_{3k+4},r_{3k+4}]\leftarrow[\ell_{3k+1}-3\heps,p_{3k+3}+\heps]$
            }
            \ELSE{
                \STATE $[\ell_{3k+4},r_{3k+4}]\leftarrow[p_{3k+3}-\heps,r_{3k+1}+3\heps]$
            }
            \ENDIF
            \IF{
            $s>\log^3 T$ and
            $\heps>\frac{1}{T}$}{
                \STATE $\heps\leftarrow \heps/2$
            }
            \ENDIF
        }
        \ELSE{
            \STATE $\heps\leftarrow 2\heps$
            \STATE $[\ell_{3k+4},r_{3k+4}]\leftarrow [\ell_{3a+1}-(3k+3-3a)\heps,r_{3a+1}+(3k+3-3a)\heps]$
        }
        \ENDIF
    }
    \ENDFOR
\end{algorithmic}
\end{algorithm}
The same as Algorithm~\ref{alg:adv-unknownfixedeps-symmetric} for unknown fixed changing rate, the algorithm takes a round of 3 time steps each time, and always keeps a confidence interval $[\ell_t,r_t]$ for the true value of the buyer at each time step. The same as in Algorithm~\ref{alg:adv-knownfixedeps-symmetric}, for three consecutive time steps $3k+1,3k+2,3k+3$ in round $k$, the algorithm first prices at the bottom of the confidence interval, then at the top of the confidence interval, finally in the middle of the interval. If $\heps\geq\eps_{3k+1}$, this is going to be a ``good round'' with $v_{3k+1}\geq p_{3k+1}$ and $v_{3k+2}<p_{3k+2}$, so no bad evidence is found by the algorithm. If this holds for $\log^3 T$ consecutive rounds, it means that the algorithm has been stable with the current $\heps$, and the confidence interval has been $O(\heps_{3k+1})$ for $\Omega(\log^3T)$ steps. Then the algorithm decreases $\heps$ by a factor of $\frac{1}{2}$ to try to explore the possibility that the true changing rate $\eps_{3k+1}$ is currently much smaller than $\heps$. However, when $\heps<\eps_{3k+1}$, the confidence interval may expand not enough to contain the true value of the buyer, so there is possibility that this round becomes a ``bad round'' with $v_{3k+1}<p_{3k+1}$ or $v_{3k+2}\geq p_{3k+2}$. In this case, the algorithm doubles $\heps$. The confidence bound is not accurate as $\heps<\eps_{3k+1}$ previously. However, if the new $\heps\geq\eps_{3k+1}$, the algorithm gets the correct confidence bound $[\ell_{3a+1},r_{3a+1}]$ in the previous good round $a$, and calculates a correct confidence bound $[\ell_{3k+4},r_{3k+4}]\leftarrow [\ell_{3a+1}-(3k+3-3a)\heps,r_{3a+1}+(3k+3-3a)\heps]$ for time step $3k+4$, since the buyer's value cannot change more than $(3k+3-3a)\eps_{3a+1}<(3k+3-3a)\heps$ in $3k+3-3a$ steps.

Now we analyze the performance of the algorithm. Partition time horizon $[T]$ to $\log T$ intervals $I_1,I_2,\cdots,I_{\log T}$, such that for each time interval $I_i$ and time $t\in I_i$, $\eps_t\in(2^{-i},2^{-i+1}]$. In other words, each time interval contains time steps with similar changing rate. 
For any time interval $I_i$, let $\eps_i^*=2^{-i+1}$. We argue that in time interval $I_i$ the symmetric loss is $O(\log^4T+|I_i|\eps_i^*=O(\log^4T+\sum_{t\in I_i}\eps_t)$, and the theorem follows immediately by taking the average of the symmetric loss in all time intervals.

Notice that in $I_i$, the algorithm may start with $\heps>\eps_i^*$, then decrease gradually to reach $\heps=\eps_i^*$. The symmetric loss in this process is $O(\log^4T)$, since there will be no bad rounds, and $\heps$ will stay at each $\heps>\eps_i^*$ for at most $O(\log^3 T)$ steps, thus reduce to $\eps_i^*$ in $O(\log^4T)$ steps. Now we analyze what happens when $\heps\leq \eps_i^*$.


For any bad round $k$, since it is at most $b$ rounds from the closest good round $a=k-b$ with confidence interval $[\ell_{3a+1},r_{3a+1}]$, all of the values and prices in the round are bounded by interval $[\ell_{3a+1}-3b\eps_i^*,r_{3a+1}+3b\eps_i^*]\subseteq[\ell_{3a+1}-3\log T\eps_i^*,r_{3a+1}+3\log T\eps_i^*]$, since in round $a$ all of the values and confidence boundaries are bounded by $[\ell_{3a+1}-O(1)\eps_i^*,r_{3a+1}+O(1)\eps_i^*]$. Thus the symmetric loss for any bad round that is $b$ rounds away from a good round is $O(r_{3a+1}-\ell_{3a+1})+O(\eps_i^*)+6\log T\eps_i^*=O(\log T\eps_i^*)$.

As analyzed in Algorithm~\ref{alg:adv-unknownfixedeps-symmetric} for static changing rate, in a good round $k$ without bad evidence, the symmetric loss in the round is proportional the length of the confidence interval plus the true changing rate $\eps_{3k+1}$ and the estimated changing rate $\heps_k$ for the round, i.e. $O(r_{3k+1}-\ell_{3k+1}+\heps_k+\eps_{3k+1})$. In such a good round, since $\heps_k$, $\eps_{3k+1}$ and the length of confidence interval $r_{3k+1}-\ell_{3k+1}$ are all $O(\eps_i^*)$, the good round has symmetric loss $O(\eps_i^*)\leq c\eps_i^*$ for some constant $c$. 

The same as in the analysis of Algorithm~\ref{alg:adv-decreasingeps-pricing} for the revenue loss setting, we can use the loss of good rounds to compensate the loss of the bad rounds. We argue that the per-step symmetric loss in $I_i$ is $\leq 2c\eps_i^*$. Between two blocks of consecutive bad rounds, there are at least $O(\log^3 T)$ good rounds. In a block of bad rounds, the symmetric loss is $O(\log T\eps_i^*)$ more than the expected $2c\eps_i^*$ loss per step, thus $O(\log^2 T\eps_i^*)$ more in total in the block of bad rounds. At the same time, In the $O(\log^3 T)$ consecutive good rounds, the symmetric loss is $c\eps_i^*$ less than the expected $2c\eps_i^*$ loss per step, thus $O(\log^3 T\eps_i^*)$ less in total in the block of good rounds.

Combine the cases above for good rounds and bad rounds, we know that in time interval $I_i$ with $\heps\leq\eps_i^*$, the symmetric loss is $\leq c\eps_i^*$ in each round on average. This finishes the proof of the total symmetric loss in the time interval being $O(\log^4T+|I_i|\eps_i^*=O(\log^4T+\sum_{t\in I_i}\eps_t)$.

Thus the total symmetric loss of the algorithm is $\sum_{i\leq \log T}O(\log^4T+\sum_{t\in I_i}\eps_t)=To(1)+O(1)\sum_{t\in I_i}\eps_t=T\tilde{O}(\bar{\eps})$.

\end{proof}

For the revenue loss minimization problem for an adversarial buyer, we can also recover the results in previous sections, even when $\eps_t$ are unknown to the seller. We describe the result and the algorithm in detail here.

\begin{theorem}\label{thm:adv-decreasingeps-pricing}
If the buyer has adversarial values and non-increasing changing rates $\eps_t$ unknown to the seller, there exists an online pricing algorithm with revenue loss $\tilde{O}(\bar{\e}^{1/2})$, here $\bar{\e}=\frac{1}{T}\sum_{t=1}^{T}\eps_t$. Further, no online algorithm can obtain a revenue loss less than $\Omega(\bar{\e}^{1/2})$.
\end{theorem}

\begin{proof}
The $\Omega(\bar{\e}^{1/2})$ tightness result has been shown in Theorem~\ref{thm:adv-knownfixedeps-pricing} with all $\eps_t$ being identical. Now we show that there exists an algorithm with revenue loss $\tilde{O}(\bar{\e}^{1/2})$.

When $\eps_t$ decreases, the seller needs to detect such a trend timely, otherwise the loss of each time step is going to be not comparable to $\sum_t\eps_t$. 
We propose the following algorithm for the seller, that repeatedly guesses the current level of changing rate at each time step. The algorithm starts with guessing $\heps=\frac{1}{2}$ being an estimate of $\eps_t$, and reduces the value of the guess $\heps$ by a factor of $\frac{1}{2}$ if in several time steps the algorithm cannot find any evidence supporting $\heps<\eps_t$. Whenever the algorithm finds evidence that supports $\heps<\eps_t$, the algorithm repeatedly doubles $\heps$ and updates the confidence interval according to the new $\heps$, until the evidence of $\heps<\eps_t$ disappears. Such a dynamic update of $\heps$ keeps the 
revenue loss bounded. 

To be more specific, $\heps$ decreases by a factor of $\frac{1}{2}$ if the seller has not observed any evidence of $\heps<\eps_t$ for \textit{long enough time}. In particular, the algorithm tries to run $\heps^{-1/2}$ identical phases in Algorithm~\ref{alg:adv-unknownfixedeps-pricing}, and will halve $\heps$ when the buyer passes all checking steps. The algorithm is described in Algorithm~\ref{alg:adv-decreasingeps-pricing}.

\begin{algorithm}[htb]
\caption{Revenue loss minimizing algorithm for adversarial buyer with unknown decreasing changing rate $\eps_t$}
\label{alg:adv-decreasingeps-pricing}
\begin{algorithmic}
\STATE $\heps\leftarrow \frac{1}{2}$
\WHILE{true}{
 \FOR{$\heps^{-1/2}$ phases of length $m_{\heps}=\heps^{-1/2}$}{
    \STATE At the beginning of phase $[t_0+1,t_0+m_{\heps}]$, apply Algorithm~\ref{alg:binarysearch} to locate the current value of the buyer in an interval $[\ell_t,r_t]$ with length $\sqrt{\heps}$
    \STATE Randomly select a $t^*\in[t_0+1,t_0+m_{\heps}]$
    \FOR{each step $t$ in the phase}{
        \IF{$t\neq t^*$}{
            \STATE Price at $p_t\leftarrow \ell_t$
            \IF{$v_t<p_t$}{
                \STATE $\heps\leftarrow 2\heps$ and go back to the beginning of the while loop (terminate the $\heps^{-1/2}$ phases)
            }
            \ENDIF
        }
        \ELSE{
            \STATE Price at $p_t\leftarrow r_t$
            \IF{$v_t\geq p_t$}{
                \STATE $\heps\leftarrow 2\heps$ and go back to the beginning of the while loop (terminate the $\heps^{-1/2}$ phases)
            }
            \ENDIF
        }
        \ENDIF
        \STATE $[\ell_{t+1},r_{t+1}]\leftarrow[\ell_{t}-\eps,r_t+\eps]$
    }
    \ENDFOR
 }
 \ENDFOR
 \STATE $\heps\leftarrow \heps/2$ if $\heps>\frac{1}{T}$
 }
 \ENDWHILE
\end{algorithmic}
\end{algorithm}

Now we analyze the performance of the algorithm. Partition the time horizon $[T]$ to $\log T$ intervals $I_1,\cdots,I_{\log T}$, such that for each time interval $I_i$ and time $t\in I_i$, $\eps_t\in (2^{-i},2^{-i+1}]$. Let $\eps_i^*=2^{-i+1}$. We argue that in time interval $I_i$ the total revenue loss is $\tilde{O}((\eps_i^*)^{-1/2}+|I_i|(\eps_i^*)^{1/2})$.

In interval $I_i$ the algorithm may start with $\heps>\eps_i^*$, and then gradually decreases to reach $\heps=\eps_i^*$ and never become larger than $\eps_i^*$ later. In this process, the revenue loss is $\tilde{O}(1)$ in each phase as shown in the proof of Theorem~\ref{thm:adv-unknownfixedeps-pricing}, thus $\tilde{O}(\heps^{-1/2})$ loss for every value $\heps>\eps_i^*$ and at most $\sum_{\heps>\eps_i^*}\tilde{O}(\heps^{-1/2})=\tilde{O}((\eps_i^*)^{-1/2})$ in total.

Now we show that after $\heps$ reaches $\eps_i^*$, the revenue loss is $\tilde{O}((\eps_i^*)^{1/2})$ per step on average.
First we study the revenue loss from each $\heps^{-1/2}$ phases with changing rate $\heps$.  The same as in previous proofs, a piece of ``bad evidence'', or a piece of evidence of $\heps<\eps$ is the event of $v_{t}<p_t$ in a non-checking step $t\neq t^*$ or $v_t\geq p_t$ in a checking step $t=t^*$. If no evidence of $\heps<\eps$ is detected, then the revenue loss is at most some constant $c=\tilde{O}(1)$ in each phase, thus $c\heps^{-1/2}$ for the $\heps^{-1/2}$ phases with $\heps^{-1}$ time steps. We also observe that if we run the algorithm with changing rate $\eps_i^*$, the revenue loss of such $\heps^{-1}$ time steps is going to be $c(\eps_i^*)^{1/2}$ per step thus $c\heps^{-1}(\eps_i^*)^{1/2}$ in total. 

We argue that the per-step revenue loss in $I_i$ is at most $2c(\eps_i^*)^{1/2}$. In $\heps^{-1/2}$ phases where no bad evidence is found, the algorithm actually gets $c(\heps^{-1}(2\eps_i^*)^{1/2}-\heps^{1/2})>c\heps^{-1}(\eps_i^{*})^{1/2}>c\heps^{-1/2}$ less revenue loss than the expected benchmark ($2c(\eps_i^*)^{1/2}$ per step). In any phase with estimated changing rate $\heps$ where a piece of bad evidence is found, as shown in the analysis of Algorithm~\ref{alg:adv-unknownfixedeps-pricing}, in expectation $m_{\heps}=\heps^{-1/2}$ steps with value out of confidence bound has occurred, and contributes at most $\heps^{-1/2}$ total additional revenue loss more than the normal $2c(\eps_i^*)^{1/2}$ loss per step. Therefore, every time the algorithm goes through $\heps^{-1/2}$ phases without bad evidence, the algorithm has at least $c\heps^{-1/2}$ less revenue loss than expected; every time the algorithm with estimated changing rate $\heps$ finds a phase with a piece of bad evidence, the algorithm has at most $c\heps^{-1/2}$ more revenue loss than expected. Observe that in each iteration with $\heps$ decreases by $\frac{1}{2}$ no bad evidence is detected, and bad evidence is found in each iteration with $\heps$ getting doubled. Thus the number of iterations with no bad evidence being detected is at least the number of iterations with bad evidence found, which means that the algorithm has no more revenue loss than the expected $2c(\eps_i^*)^{1/2}$ per step. To summarize, in $I_i$ after $\heps$ reaches $\eps_i^*$, the revenue loss is $\tilde{O}((\eps_i^*)^{1/2})$ per step.

Above reasoning shows that in each time interval $I_i$, the total revenue loss is $\tilde{O}((\eps_i^*)^{-1/2}+|I_i|(\eps_i^*)^{1/2})$. Sum up over all $i$, the total revenue loss of all time steps is 
\begin{eqnarray*}
& &\sum_{i\leq\log T}\tilde{O}((\eps_i^*)^{-1/2}+|I_i|(\eps_i^*)^{1/2})\\
&=&\tilde{O}(T^{1/2})+\tilde{O}(1)\sum_{i}|I_i|(\eps_i^*)^{1/2}\\
&\leq&To(1)+\tilde{O}(1)\sum_{i}|I_i|\bar{\e}^{1/2}=\tilde{O}(T\bar{\e}^{1/2}),
\end{eqnarray*}
Here the inequality is by Cauchy-Schwarz. Thus the average revenue loss of each time step is $\tilde{O}(\bar{\e}^{1/2})$.

\end{proof}

Such a result can also be extended for a stochastic buyer.

\begin{theorem}\label{thm:rand-decreasingeps-pricing}
If the buyer has stochastic value and non-increasing changing rate $\epsilon_t$ unknown to the seller, then there exists an online pricing algorithm with revenue loss $\tilde{O}(\bar{\e}^{2/3})$, here $\bar{\e}=\frac{1}{T}\sum_{t=1}^{T}\eps_t$. Furthermore, no online algorithm can obtain a revenue loss less than $\Omega(\bar{\e}^{2/3})$.
\end{theorem}

\begin{proof}
The $\Omega(\bar{\e}^{2/3})$ tightness follows from Theorem~\ref{thm:rand-knownfixedeps-pricing}. The algorithm for $\tilde{O}(\bar{\e}^{2/3})$ revenue loss can be obtained by combining Algorithm~\ref{alg:rand-unknownfixedeps-pricing} for stochastic buyer with unknown fixed changing rate,  and Algorithm~\ref{alg:adv-decreasingeps-pricing} for adversarial buyer with unknown decreasing changing rate. In particular, we replace each iteration of fixed $\heps$ with $\heps^{-1/2}$ phases of length $\heps^{-1/2}$, by $\heps^{-2/3}$ phases of length $\heps^{-2/3}$ in Algorithm~\ref{alg:rand-unknownfixedeps-pricing}.

The analysis of the algorithm is identical to the analysis of Algorithm~\ref{alg:adv-decreasingeps-pricing}. Partition the time horizon $[T]$ to intervals $I_1,\cdots,I_{\log T/2}$, each with $\eps_t\in(2^{-i},2^{-i+1}=\eps_i^*]$ for any $t\in I_i$. We can prove that the total revenue loss in each time interval is $\tilde{O}((\eps_i^*)^{-2/3}+|I_i|(\eps_i^*)^{2/3})$. The total revenue loss of the algorithm is
\begin{eqnarray*}
& &\sum_{i\leq\log T}\tilde{O}((\eps_i^*)^{-2/3}+|I_i|(\eps_i^*)^{2/3})\\
&=&\tilde{O}(T^{2/3})+\tilde{O}(1)\sum_{i}|I_i|(\eps_i^*)^{2/3}\\
&\leq&To(1)+\tilde{O}(1)\sum_{i}|I_i|\left(\frac{\sum_i |I_i|\eps_i^*}{\sum_i |I_i|}\right)^{2/3}\\
&=&\tilde{O}(T\bar{\e}^{2/3}),
\end{eqnarray*}
Here the inequality follows from H\"{o}lder's inequality and $\bar{\e}=\frac{1}{T}\sum_t\eps_t$. Thus the average revenue loss of each time step is $\tilde{O}(\bar{\e}^{2/3})$.
\end{proof}

For the revenue loss minimization problem, it is hard to obtain positive results when the changing rates $\eps_t$ are arbitrary, since setting a price slightly higher than the true value in a step can result in a huge revenue loss. Surprisingly, even if $\eps_t$ for each time step can change arbitrarily, we can still achieve the $\tilde{O}(\bar{\eps})$ loss in previous sections, only losing a tiny $O(\log T)$ factor.

\begin{theorem}\label{thm:adv-unknownchangingeps-symmetric}

If the buyer has adversarial values and dynamic changing rate $\epsilon_t$ unknown to the seller, there exists an online pricing algorithm with symmetric loss $\tilde{O}(\bar{\eps} \log T)$ for $\bar{\eps}=\frac{1}{T}\sum_{t\in [T]}\eps_t$. Further, no online algorithm can obtain a symmetric loss less than $\Omega(\bar{\eps})$.

\end{theorem}

\begin{proof}

Suppose for a moment that the algorithm is allowed to set multiple pricing queries for a single time step. The algorithm maintains a correct confidence interval $[\ell_{t},r_{t}]$ that contains the value $v_t$ of the buyer at each time step. At time $t+1$, the seller does not know the exact value change $v_{t+1}-v_t$. Furthermore, she also does not know a bound of the value change $\eps_t\geq|v_{t+1}-v_t|$. However, the seller can try to price at $\ell_t-\delta_j$ and $r_t+\delta_j$ repeatedly for every $j$ and $\delta_j=2^{j}T^{-1}$. When $j$ has increased such that the algorithm finds that $\ell_t-\delta_j<v_{t+1}<r_t+\delta_j$, the seller can then price at $\frac{\ell_t+r_t}{2}$ to get a new correct confidence interval $[\ell_{t+1},r_{t+1}]\leftarrow[\ell_t-\delta_j,\frac{\ell_t+r_t}{2}]$ or $[\frac{\ell_t+r_t}{2},r_t+\delta_j]$. Let $a_t=r_t-\ell_t$ be the length of the confidence interval at time $t$, then $a_{t+1}=\frac{1}{2}a_t+\delta_j$. Since $v_{t+1}$ and all prices at this time step are in $[\ell_t-\delta_j,r_t+\delta_j]$, thus the symmetric loss of each query is at most $2a_t+2\delta_j=4a_{t+1}$. Also notice that $v_{t+1}\not\in[\ell_t-\delta_j/2,r_t+\delta_j/2]$ since the algorithm does not stop at $j-1$ at time $t$, thus $\eps_t>\frac{\delta_j}{2}$ as the buyer's value must has changed by $\frac{\delta_j}{2}$ at this step. Therefore the symmetric loss of each time step $t+1$ is at most $4a_{t+1}$, while $a_{t+1}\leq \frac{1}{2}a_t+\eps_t$. Then the total symmetric loss of the algorithm is $\leq 4\sum_{t\in[T]}a_t\leq 8\sum_{t\in [T]}\eps_t$, which implies the average symmetric loss is $O(\bar{\eps})$.

However, we are not allowed to have multiple pricing queries for the same value. The key observation to be proved in this section is that when we serialize the pricing queries in such an algorithm with at most $k$ queries per step, the symmetric loss only increases by a factor of $O(k)$. Since the above algorithm has at most $O(\log T)=\tilde{O}(1)$ pricing queries in each step, the serialized algorithm's symmetric loss only increases by $\tilde{O}(1)$.

The serialized algorithm runs as follows. The algorithm runs in phases, with each phase $i$ having time steps $[t_i+1,t_{i+1}]$ for a to-be-determined stopping time $t_{i+1}$. In each phase, the algorithm starts with an estimated confidence interval $[\ell_{t_i+1},r_{t_i+1}]$, and repeatedly set prices $\ell_{t_i+1}-\delta_j$ and $r_{t_i+1}+\delta_j$ for every two time steps. If the algorithm observes $v_{t_i+2j+1}\geq \ell_{t_i+1}-\delta_j$ and $v_{t_i+2j+2}<r_{t_i+1}+\delta_j$, the phase stops at time $t_{i+1}=t_i+2j+3$ with price $p_{t_{i+1}}=\frac{\ell_{t_i+1}+r_{t_i+1}}{2}$ trying to halve the length of the confidence interval.

\begin{algorithm}[htbp]
\caption{Symmetric loss minimizing algorithm for adversarial buyer with unknown dynamic changing rate $\eps_t$}
\label{alg:thm:adv-unknownchangingeps-symmetric}
\begin{algorithmic}
 \STATE $[\ell_1,r_1]\leftarrow [0,1]$
 \STATE Let $t_1+1=1$ be the starting time of the first phase
 \FOR{each phase $i$ of time interval $[t_i+1,t_{i+1}]$ with to-be-determined stopping time $t_{i+1}$}{
    \STATE Let $t_i+1$ be the starting time step of the phase
    \FOR{each two consecutive time steps $t_i+2j+1,t_i+2j+2$ with $j\geq 0$}{
        \STATE $\heps\leftarrow\delta_j=2^jT^{-1}$
        \STATE Set price $p_{t_i+2j+1}\leftarrow \ell_{t_i+1}-\delta_j$
        \STATE Set price $p_{t_i+2j+2}\leftarrow r_{t_i+1}+\delta_j$
        \IF{the seller finds $v_{t_i+2j+1}\geq p_{t_i+2j+1}$ and $v_{t_i+2j+2}< p_{t_i+2j+2}$}{
            \STATE break (from this for loop)
        }
        \ENDIF
    }
    \ENDFOR
    \STATE $t_{i+1}\leftarrow t_i+2j+3$
    \STATE Set $p_{t_{i+1}}\leftarrow \frac{\ell_{t_i+1}+r_{t_i+1}}{2}$
    \IF{$v_{t_{i+1}}<p_{t_{i+1}}$}{
        \STATE $[\ell_{t_{i+1}+1},r_{t_{i+1}+1}]\leftarrow [\ell_{t_i+1}-\delta_j,p_{t_{i+1}}]$
    }
    \ELSE{
        \STATE $[\ell_{t_{i+1}+1},r_{t_{i+1}+1}]\leftarrow [p_{t_{i+1}},r_{t_i+1}+\delta_j]$
    }
    \ENDIF
 }
 \ENDFOR
\end{algorithmic}
\end{algorithm}

Now we analyze the symmetric loss of each phase $I_i=[t_i+1,t_{i+1}]$. The first observation is that at each time step, the true value may not be in the confidence interval, due to the delay of prices. The confidence interval at time $t_i+1$ is $[\ell_{t_i+1},r_{t_i+1}]$ in the algorithm, but since the value of $\ell_{t_i+1}$ and $r_{t_i+1}$ are determined in the previous three time steps, thus the ``true'' confidence interval can be defined by $v_{t_i+1}\in [\ell_{t_i+1}-\eps_{t_i-2}-\eps_{t_i-1}-\eps_{t_i},r_{t_i+1}+\eps_{t_i-2}+\eps_{t_i-1}+\eps_{t_i}]$. Furthermore, for any time $t'\in I_i$, $v_{t'}\in [\ell_{t_i+1}-\sum_{t_i-2\leq t< t_{i+1}}\eps_t,r_{t_i+1}+\sum_{t_i-2\leq t<t_{i+1}}\eps_t]$. 

Observe that the algorithm in phase $i$ breaks at $j$,
thus $v_{t_i+2j+1}\geq p_{t_i+2j+1}=\ell_{t_i+1}-\delta_{j}$ and $v_{t_i+2j+2}<p_{t_i+2j+2}=r_{t_i+1}+\delta_{j}$. Thus any price in the phase is in range $[\ell_{t_i+1}-\delta_{j},\ell_{t_i+1}+\delta_{j}]$, while any value $v_t$ in the phase is in range $[v_{t_i+2j+1}-\sum_{t\in I_i}\eps_t,v_{t_i+2j+2}+\sum_{t\in I_i}\eps_t]\subseteq[\ell_{t_i+1}-\delta_{j}-\sum_{t\in I_i}\eps_t,r_{t_i+1}+\delta_{j}+\sum_{t\in I_i}\eps_t]$, so the symmetric loss at any time step in the phase is at most $r_{t_i+1}-\ell_{t_i+1}+2\delta_j+\sum_{t\in I_i}\eps_t$. 

Also since the algorithm does not break at $j-1$ in the for loop, thus either $v_{t_i+2j-1}<p_{t_i+2j-1}=\ell_{t_i+1}-\delta_{j-1}$, or $v_{t_i+2j}>p_{t_i+2j}=r_{t_i+1}+\delta_{j-1}$. In either case, $\sum_{t_i-2\leq t<t_{i+1}}\eps_t>\delta_{j-1}=\frac{1}{2}\delta_j$.

Let $a_i=r_{t_i+1}-\ell_{t_i+1}$ be the length of the confidence interval at the beginning of phase $i$, and $b_i=\sum_{t\in I_i}\eps_t$ be the total changing rate of the value in phase $i$. Then according to the binary search step, $a_{i+1}=\frac{1}{2}a_i+\delta_{j_i}$ where $j_i$ is the value of $j$ at the end of phase $i$. Since $\delta_{j_i}=2\delta_{j_i-1}<\sum_{t_i-2\leq t<t_{i+1}}\eps_t<b_{i-1}+b_i$, we have
\begin{equation*}
    a_{i+1}<\frac{1}{2}a_i+2(b_i+b_{i-1}).
\end{equation*}
Then $a_{i+1}$ can be represented using only $b_i$'s as follows:
\begin{equation}\label{eqn:arbitrary-recursion}
    a_{i+1}<2(b_i+b_{i-1})+(b_{i-1}+b_{i-2})+\frac{1}{2}(b_{i-2}+b_{i-3})+\cdots
\end{equation}

Since the total symmetric loss in phase $i$ is upper bounded by $a_{i+1}+\sum_{t\in I_i}\eps_t$ per step thus $|I_i|(a_{i+1}+b_i)\leq (2\log T+1)(a_{i+1}+b_i)$ in total, summing over all phases we get the total symmetric loss of the algorithm for all phases is
$L:=(2\log T+1)\sum_i(a_{i+1}+b_i)$. $L$ can be represented with only $b_i$ by applying \eqref{eqn:arbitrary-recursion} to $L$. Since the coefficient of $b_i$ in L is $O(\log T)$ for each $b_i$, thus $L=O(\log T)\sum_i b_i=O(\log T)\sum_{t\in [T]}\eps_t$. Therefore the average symmetric loss of the algorithm is $\tilde{O}(\bar\eps)$.

\end{proof}


\bibliographystyle{plainnat}
\bibliography{pricing}

\appendix

\section{Comparison to Fixed Price Benchmark}
\begin{theorem}\label{thm:beatkl}
The no-regret algorithm for an adversarial buyer in \cite{kleinberg2003value} can have $\Omega(1)$ revenue loss.
\end{theorem}

\begin{proof}
Consider the following sequence of buyer values $\mathbf{v}=(v_1,\cdots,v_T)$ with changing rate $\eps$. Let $m=\frac{1}{\eps}$. For any integer $k\geq 0$ and $1\leq j\leq m$, let $v_{2km+j}=j\eps$, and $v_{2km+m+j}=1-(j-1)\eps$. In other words, the value sequence is periodic with period $2m$; in each phase of repetition, $v_t$ starts with $\eps$ and increases by $\eps$ for each step until it reaches 1, then starts with 1 and decreases by $\eps$ for each step. 

We show that the zero-regret pricing algorithm in \cite{kleinberg2003value} gives $\Omega(T)$ revenue loss, if $T>m^3\ln m=\frac{1}{\eps^3}\ln\frac{1}{\eps}$. The algorithm applies the seminal bandit algorithm EXP3 \cite{auer1995gambling} with each arm $1\leq i\leq m$ corresponding to a price $p_i=i\eps$ which is a multiple of $\eps$. We describe the algorithm in detail as follows.

\begin{algorithm}[H]
\caption{No-regret pricing algorithm in \cite{kleinberg2003value}}
\label{alg:kleinberg-leighton}
\begin{algorithmic}
 \STATE $\eta\leftarrow \sqrt{\frac{\ln m}{Tm}}$
 \STATE $w_1(i)\leftarrow 1$ for each arm $1\leq i\leq m$
 \FOR{each time step $t$}{
     \STATE $q_t(i)\leftarrow (1-\eta)\frac{w_t(i)}{W_t}+\frac{\eta}{m}$
     \STATE Draw $i_t$ from $\mathbf{q}_t$, i.e. $\Pr[i_t=i]=q_t(i)$ for each $1\leq i\leq m$
     \STATE Price at $p_t\leftarrow i_t\eps$, and observe revenue gain $r_t\leftarrow p_t\one \{ v_t \geq p_t \}$
     \STATE Update $w_{t+1}(i)\leftarrow w_{t}(i)$ for $i\neq i_t$, and $w_{t+1}(i_t)\leftarrow w_{t}(i_t)\exp\left(\eta\frac{r_t}{mq_t(i_t)}\right)$
 }
 \ENDFOR
\end{algorithmic}
\end{algorithm}

It suffices to show that the revenue loss in each time interval $I=[2mk+1,2mk+m]$ is $\Omega(m)$ in expectation. Partition the arms to two groups $S_1=\{1,2,\cdots,\frac{m}{2}\}$ and $S_2=\{\frac{m}{2}+1,\cdots,m\}$. Then we have the following observations.

\paragraph{Observation 1: with probability $>\frac{1}{e}$, $q_t(i_t)\geq\frac{1}{m^2}$ for each time step $t\in I$.} This is because the probability of selecting an arm $i$ with $q_t(i)<\frac{1}{m^2}$ is at most $\frac{m-1}{m^2}$ since there are at most $m-1$ such arms, thus the probability of $q_t(i_t)\geq\frac{1}{m^2}$ for each $t\in I$ is at least $(1-\frac{m-1}{m^2})^m>\frac{1}{e}$.

\paragraph{Observation 2: with probability $>\frac{1}{e}$, $w_{t}(i)$ increases to a factor of $e$ in the entire time interval $I$ for every arm $1\leq i\leq m$.} This is because if $q_t(i_t)\geq\frac{1}{m^2}$, $w_{t+1}(i_t)$ increases to a fraction of $\exp\left(\eta\frac{r_t}{mq_t(i_t)}\right)\leq \exp\left(\sqrt{\frac{\ln m}{Tm}}\frac{1}{m\cdot m^{-2}}\right)<e^{1/m}$ compared to $w_{t}(i_t)$ for each time step $t$ since $T>m^3\ln m$. Thus if $q_t(i_t)\geq\frac{1}{m^2}$ for each time step $t\in I$, then for any arm $1\leq i\leq m$ and time step $t\in I$,  $w_t(i)<ew_{2mk+1}(i)$ as there are only $m$ steps in $I$. Then Observation 2 follows from Observation 1.

Then we are ready to argue that in time interval $I$, the total revenue loss of Algorithm~\ref{alg:kleinberg-leighton} is $\Omega(m)$. We prove a stronger result, that the algorithm has $\Omega(m)$ revenue loss in either time interval $I_1=[2mk+\frac{m}{4}+1,2mk+\frac{m}{2}]$ or $I_2=[2mk+\frac{3m}{4}+1,2mk+m]$. Let $W_{t,1}=\sum_{i\in S_1}w_{t}(i)$ be the total weight of arms in $S_1=\{1,2,\cdots,\frac{m}{2}\}$ at time step $t\in I$, and similarly define $W_{t,2}=\sum_{i\in S_2}w_{t}(i)$ for arms $S_2$. If $W_{2mk+1,1}\geq W_{2mk+1,2}$ at the beginning of $I$, by Observation 2 $W_{t,1}\geq \frac{1}{e}W_{t,2}$ for any time step $t\in I$. By definition of $q_t(i)$, with constant probability $>\frac{1/e}{1+1/e}=\frac{1}{e+1}$, the algorithm selects an arm $i\in S_1$ at any time step $t\in I$. In particular, with probability $\frac{1}{e+1}$, an arm $i_t\in S_1$ (which corresponds a price $p_t<\frac{1}{2}$) is selected in time step $t\in I_2$. Since in time step $t\in I_2$ we always have $v_t>\frac{3}{4}$, thus the revenue loss is at least $\frac{1}{4}$ for each time step $t\in I_2$ with constant probability, thus $\Omega(m)$ in expectation in total. Similarly, if $W_{2mk+1,1}
<W_{2mk+1,2}$ at the beginning of $I$, then at each time step $t\in I_1$ we have $\frac{1}{4}<v_{t}\leq\frac{1}{2}$, but with constant probability an arm in $S_2$ is selected and results in $p_t>\frac{1}{2}$ and revenue loss $v_t>\frac{1}{4}$, thus $\Omega(m)$ in expectation in total. To summarize, in either  $I_1=[2mk+\frac{m}{4}+1,2mk+\frac{m}{2}]$ or $I_2=[2mk+\frac{3m}{4}+1,2mk+m]$, the total revenue loss is $\Omega(m)$ in expectation. Summing up the revenue loss for all phases, the total revenue loss for Algorithm~\ref{alg:kleinberg-leighton} is $\Omega(T)$ for all time steps, thus $\Omega(1)$ on average.

\end{proof}

\section{Buyer's Changing Rate is Dynamic, Known}
In this section, we study the problem where $\epsilon_t$ are publicly known by the seller. In this case, the algorithms are much simpler than the case in the main body where $\eps_t$ are unknown to the seller, and we are able to handle the case where $\eps_t$ are not decreasing for the revenue loss.

\begin{theorem}\label{thm:adv-knownchangingeps-symmetric}
If the buyer has adversarial values and changing rates $\eps_t$ known to the seller, then there exists an online pricing algorithm with symmetric loss $O(\bar{\eps})=O(\frac{1}{T}\sum_{t=1}^{T}\eps_t)$. Furthermore, no online algorithm can obtain a symmetric loss less than $\Omega(\bar{\eps})$.
\end{theorem}

\begin{proof}
The tightness result $\Omega(\bar{\eps})$ was shown in Theorem~\ref{thm:adv-knownfixedeps-symmetric} even when $\eps_t$ are identical.
Now we show that a binary search algorithm that tracks the confidence bound gets symmetric loss $O(\bar{\eps})$. We only slightly modify Algorithm~\ref{alg:adv-knownfixedeps-symmetric} as follows. The algorithm maintains confidence bound $[\ell_t,r_t]$ at each time step, and price at $p_t=\frac{\ell_t+r_t}{2}$. If $v_t<p_t$, then $[\ell_{t+1},r_{t+1}]=[\ell_t-\eps_t,p_t+\eps_t]$; if $v_t\geq p_t$, then $[\ell_{t+1},r_{t+1}]=[p_t-\eps_t,r_t+\eps_t]$. 

Let $a_t=r_t-\ell_t$ be the length of the confidence interval at each time. The symmetric loss at time $t$ is upper bounded by $\frac{a_t}{2}$. Observe that $a_{t+1}\leq\frac{a_t}{2}+2\eps_t$, we have \begin{equation}\label{eqn:at}
    a_{t+1}\leq2\eps_t+\eps_{t-1}+2^{-1}\eps_{t-2}+\cdots+2^{-(t-1)}\eps_1+2^{-(t-1)}a_1.
\end{equation}
Notice that the total symmetric loss is upper bounded by
\(\sum_{t=1}^{T}\frac{1}{2}a_t\). When we apply \eqref{eqn:at} to the formula, the coefficient of any $\eps_t$ is at most $2+1+2^{-1}+\cdots\leq 4$. Thus the symmetric loss is $O(\bar{\eps})$.

\end{proof}

\begin{theorem}\label{thm:adv-knownchangingeps-pricing}
If the buyer has adversarial values and changing rates $\epsilon_t$ known to the seller, then there exists an online pricing algorithm with revenue loss $\tilde{O}(\bar{\e}^{1/2})$, here $\bar{\e}=\frac{1}{T}\sum_{t=1}^{T}\eps_t$. Furthermore, no online algorithm can obtain a revenue loss less than $\Omega(\bar{\e}^{1/2})$.
\end{theorem}

\begin{proof}
The tightness result has been shown in Theorem~\ref{thm:adv-knownfixedeps-pricing}. Now we construct an online pricing algorithm with revenue loss $\tilde{O}(\bar{\e}^{1/2})$. The algorithm is obtained by slightly modifying Algorithm~\ref{alg:adv-knownfixedeps-pricing}. Firstly, the update rule of confidence interval is changed to $[\ell_{t+1},r_{t+1}]\leftarrow [\ell_t-\eps_t,r_t+\eps_t]$. Secondly, the length of each phase is no longer a fixed length. A phase starting at time step $t_0+1$ has length $m$ such that $\sum_{t=t_0+1}^{t=t_0+m-1}\eps_t\leq\sqrt{\bar\eps}<\sum_{t=t_0+1}^{t=t_0+m}\eps_t$. In other words, the total possible value change in each phase is $\Theta(\sqrt{\bar\eps})$, so the revenue loss of each step is $O(\sqrt{\bar\eps})$. Thirdly, in the initialization step of each phase, the algorithm uses binary search algorithm to narrow down the length of the confidence interval to $O(\bar{\eps}^{1/2})$. 

The number of phases is $O(T\bar{\eps}^{-1/2})$ since $\sum_t\eps_t=T\bar\eps$. Then the revenue loss of all phases is contributed by the loss of binary search for locating the value to a confidence interval of length $\sqrt{\bar\eps}$ in each phase, which is $O(\log\frac{1}{\sqrt{\bar\eps}})=\tilde{O}(1)$ per phase and $\tilde{O}(T\bar{\eps}^{-1/2})$ in total; plus the revenue loss of each step in the phase, which is at most $O(\sqrt{\bar\eps})$ each step thus $O(T\sqrt{\bar\eps})$ in total. Thus the total  revenue loss of the algorithm is $\tilde{O}(T\bar{\eps}^{-1/2})$ for all time steps, thus $\tilde{O}(\bar{\eps}^{-1/2})$ on average.

\end{proof}

\begin{theorem}\label{thm:rand-knownchangingeps-pricing}
If the buyer has stochastic value and changing rates $\epsilon_t$ known to the seller, then there exists an online pricing algorithm with revenue loss $\tilde{O}(\tilde{\e}^{2/3})$, here $\tilde{\e}=\sqrt{\frac{1}{T}\sum_{t=1}^{T}\eps^2_t}$. Furthermore, no online algorithm can obtain a revenue loss less than $\Omega(\tilde{\e}^{2/3})$.
\end{theorem}

\begin{proof}
The tightness result has been shown in Theorem~\ref{thm:rand-knownfixedeps-pricing}. Now we construct an online pricing algorithm with revenue loss $\tilde{O}(\te^{2/3})$. The algorithm is obtained by slightly modifying Algorithm~\ref{alg:rand-knownfixedeps-pricing}. Firstly, the update rule of confidence interval is changed to $[\ell_{t+1},r_{t+1}]\leftarrow [\ell_t-\eps_t,r_t+\eps_t]$. Secondly, the length of each phase is no longer a fixed length. A phase starting at time step $t_0+1$ has length $m$ such that $\sum_{t=t_0+1}^{t=t_0+m-1}\eps^2_t\leq\te^{4/3}<\sum_{t=t_0+1}^{t=t_0+m}\eps^2_t$. Thirdly, let $\delta=4\tilde{O}(\te^{-2/3})$, then the fixed price in the phase is $p_t=\ell_{t_0}-\delta$. Finally, in the initialization step of each phase, the algorithm uses binary search algorithm to narrow down the length of the confidence interval to $O(4\te^{2/3})$. This is achievable since for any $t$ in the phase $\eps_t\leq \te^{2/3}$, as $\sum_{t=t_0+1}^{t=t_0+m-1}\eps^2_t\leq\te^{4/3}$.  

The parameters are selected to match the analysis of Algorithm~\ref{alg:rand-knownfixedeps-pricing}. Since $\sum_t\eps_t^2=T\te^2$, there are $\frac{T\te^2}{\Omega(\te^{4/3})}=O(T\te^{2/3})$ phases. In each phase, the binary search incurs $\tilde{O}(1)$ loss, thus $\tilde{O}(T\te^{2/3})$ in all phases. Suppose that at time $t_1$ the search algorithm ends, and we get an estimate of $v_{t_1}$ with additive accuracy $4\te^{2/3}$. In the remaining of the phase, the price is fixed to be $\ell-\delta$. For any $t\in[t_1,t_0+m]$, by Azuma's inequality, the probability that $v_{t}\not\in[v_{t_1}-\delta,v_{t_1}+\delta]$ is
\begin{eqnarray*}
\Pr[v_{t}\not\in[v_{t_1}-\delta,v_{t_1}+\delta]]&\leq& 2\exp\left(-\frac{\delta^2}{2\sum_{t}\eps_t^2}\right)\\&=&2\exp\left(-\frac{\delta^2}{2\te^{4/3}}\right)<\frac{1}{m^2}.\end{eqnarray*}
By union bound, the probability that exists $t\in[t_1,t_0+m]$ such that $v_{t}\not\in[v_{t_1}-\delta,v_{t_1}+\delta]$ is at most $\frac{1}{m}$; in this case, the revenue loss in the phase is at most 1 per step thus $m$ in total, so the expected revenue loss contributed from the case is $O(1)$. Since there are $O(T\te^{2/3})$ phases, the revenue loss of this case in all phases is $\tilde{O}(T\te^{2/3})$ in total. If $v_{t}\in[v_{t_1}-\delta,v_{t_1}+\delta]$ for every $t\in[t_1,t_0+m]$, the revenue loss for setting price $\ell-\delta\geq v_{t}-4\te^{2/3}-\delta$ is at most $O(\te^{2/3}+\delta)=O(\delta)$ for each step, thus $O(T\delta)$ in total for all phases with $\ell-\delta\geq v_{t}-4\te^{2/3}-\delta$ in each time step. 
Combine all cases above, the expected revenue loss is $\tilde{O}(T\te^{2/3})$ in total, thus $\tilde{O}(\te^{2/3})$ on average.

\end{proof}

\end{document}